\documentclass[runningheads]{llncs}
\usepackage[utf8]{inputenc}
\usepackage[T1]{fontenc}

\usepackage[english]{babel}

\usepackage{amsmath, amssymb} 
\usepackage{enumitem}
\usepackage{framed}
\usepackage{tikz}
\usepackage{tikzit}
\usepackage[linesnumbered, boxed]{algorithm2e}

\tikzstyle{Species}=[fill=white, draw=black, shape=circle]
\tikzstyle{Reaction}=[fill=white, draw=black, shape=rectangle]
\tikzstyle{SCC}=[fill={rgb,255: red,191; green,191; blue,191}, draw=black, shape=circle]

\tikzstyle{Arrow}=[->]

\usepackage{caption}
\usepackage{subcaption}
\usepackage{hyperref}
\usepackage{cleveref}
\usepackage{url}
\newtheorem{notation}{Notation}

\crefname{lemme}{lemme}{lemmes}

\newcommand{\E}[1]{\mathbb{E}\left[#1\right]}

\renewcommand{\P}[1]{\mathbb{P}\{#1\}}

\def\d{\mathrm{d}}

\def\e{\mathrm{e}}

\def\l{\ell}
\def\N{\mathbb{N}}
\def\R{\mathbb{R}}

\begin{document}
\title{Graphical Conditions ensuring Equality between Differential and Mean Stochastic Dynamics}
\titlerunning{Mean stochastic behaviours}
%
\author{Hugo Buscemi\inst{1,2} \and François Fages\inst{2}\orcidID{0000-0001-5650-8266} }
\authorrunning{H. Buscemi, F. Fages.}
%
\institute{$^1$ {ENS Paris-Saclay, Univ. Paris-Saclay, France}\\
\email{Hugo.Buscemi@ens-paris-saclay.fr}\\
$^2$ {Inria Saclay, Palaiseau, France}\\
\email{Francois.Fages@inria.fr}}
\maketitle              
\begin{abstract}
  Complex systems can be advantageously modeled by formal reaction systems (RS), a.k.a.~chemical reaction networks in chemistry.
Reaction-based models can indeed be interpreted in a hierarchy of semantics, depending on the question at hand,
most notably by Ordinary Differential Equations (ODEs), Continuous Time Markov Chains (CTMCs), discrete Petri nets and asynchronous Boolean transition systems.
The last three semantics can be easily related in the framework of abstract interpretation.
The first two are classically related by
Kurtz's limit theorem which states that if reactions are density-dependent families,
then, as the volume goes to infinity, the mean reactant concentrations of the CTMC tends towards the solution of the ODE.
In the more realistic context of bounded volumes,
it is easy to show, by moment closure, that the restriction to reactions with at most one reactant ensures similarly that
the mean of the CTMC trajectories is equal to the solution of the ODE at all time points.
In this paper, we generalize that result in presence of polyreactant reactions,
by introducing the Stoichiometric Influence and Modification Graph (SIMG) of an RS,
and by showing that the equality between the two interpretations holds 
for the variables that belong to distinct SIMG ancestors of polyreactant reactions.
We illustrate this approach with several examples.
Evaluation on BioModels reveals that the condition for all variables is satisfied on models with no polymolecular reaction only.
However, our theorem can be applied selectively to certain variables of the model to provide insights into their behaviour within more complex systems.
Interestingly, we also show that the equality holds for a basic oscillatory RS implementing the sine and cosine functions of time.
\end{abstract}

\section{Introduction}

Inside cells, numerous chemical reactions can occur.
It is valuable to understand the behaviour and state of each chemical species involved in these reactions.
Such reactions can be represented by a Chemical Reaction Network (CRN),
which describes the species involved, their resulting outcomes, and the specific rates at which they react.
An Ordinary Differential Equation (ODE) can be associated with such a CRN~\cite{Feinberg77crt}.
Given a fixed initial state, each ODE has a unique deterministic solution.

On the other hand, a CRN can also be interpreted as a Continuous Time Markov Chain (CTMC)
and simulated by the Stochastic Simulation Algorithm (SSA), a.k.a.~Gillespie's Algorithm~\cite{gillespie1977exact}.
Simulations starting from the same initial state exhibit a variety of non-deterministic trajectories
and a mean trajectory.
This stochastic mean behaviour may or may not align with the deterministic solution of the ODE starting from the same initial state.
Kurtz's limit theorem~\cite{kurtz1970solutions} asserts that, under certain circumstances, the mean behaviour converges to the deterministic one
when the volume tends to the infinity.
These circumstances involve rates that are density-dependent,
possibly apart from a limited number of mode-switching transient states \cite{PH21tmcs}.

In the actual context of a cell however, the volume is bounded,
and there is also a finite, and possibly small, number of molecules for some species,
e.g.~as an extreme, only one for genes.
It is in this real-world model that we aim to study the relationship between the two continuous-time interpretations of CRNs,
by ODE and by CTMC.
Regarding the exact correspondence between the mean stochastic CTMC behaviours and the deterministic ODE behaviour, 
it is well known that if all reactions in the CRN have at most one reactant, then the equality holds.
This can be easily by the calculus of moments derived from the Chemical Master Equation\@\cite{anderson2015stochastic}.

Our contribution is an extension of the unimolecular case, while still ensuring equality.
Our result allows for polymolecular reactions with uncorrelated species, by guaranteeing the absence of correlations
through graph-theoretic conditions on the structure of the network.

For the sake of generality of these results for dynamical systems beyond the molecular interaction scale, outside of chemistry, 
we prefer to speak of formal reaction systems (RS) over mathematical variables~\cite{FGS15tcs},
rather than CRNs over molecular species.
Apart from this update of terminology, there is no difference between RSs and CRNs,
which are equipped with the same hierarchy of differential, stochastic, Petri net and Boolean semantics~\cite{FS08tcs}.

The rest of the paper is organized as follows.
The next section illustrates the discrepancy there is between the differential and stochastic interpretations
of a simple RS modeling the classical Lotka-Volterra prey-predator population dynamics system.
The following section provides the necessary preliminaries on reaction systems with mass action kinetics
and their differential and mean stochastic semantics, based respectively on ODE and CTMC interpretations.
In Sec.~\ref{sec:simg}, we introduce the Stoichiometric Influence and Modification Graph (SIMG)
which somehow combines the classical notions of reaction hypergraph and influence graph of a reaction system~\cite{FMRS18tcbb}.
That notion of SIMG is used in Sec.~\ref{sec:theorem} to express and prove our main theorem of strong equality between
the differential ODE and mean stochastic CTMC interpretations of RSs satisfying some structural conditions,
independently of the kinetics and any limit considerations.
In Sec.~\ref{sec:algo}, we provide a low complexity algorithm to test our SIMG conditions.
Then in the following sections, we provide several examples to illustrate these results.
Interestingly, we also show that strong equality property for an oscillatory RS
which implements the sine and cosine functions of time in both differential and stochastic interpretations.
Then we report on the evaluation of our condition in the repository of RS models Biomodels~\cite{BioModels2020}.
Finally, we conclude on some perspectives for future work with three open questions.

\section{Motivating Example}\label{sec:lv}

\begin{figure}[!h]
  \centering
  \includegraphics[scale=0.38]{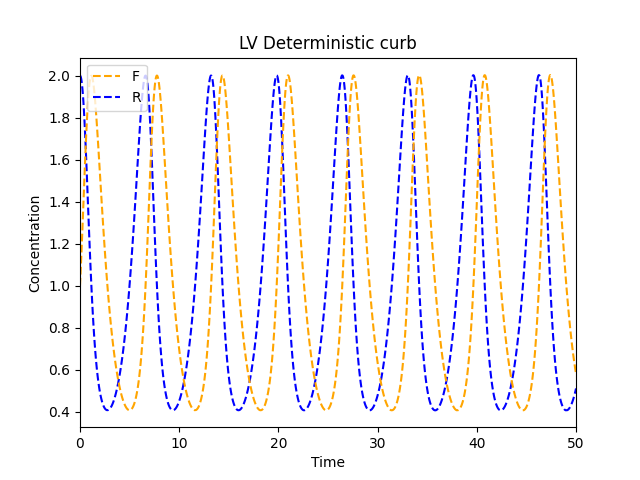}
  \caption{Deterministic ODE simulation of Lotka-Volterra reaction system.}\label{fig: LV_ODE}
\end{figure}
\begin{figure}[!h]
  \centering
  \begin{subfigure}[b]{0.49\textwidth}
    \includegraphics[width=\textwidth]{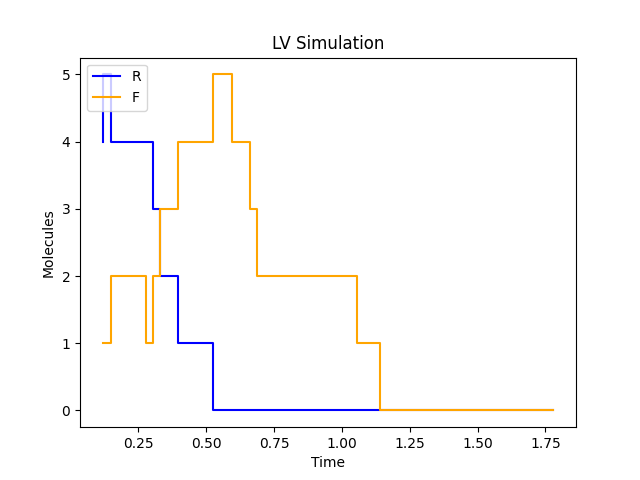}
    \caption{Standard Lotka-Volterra system}\label{fig: Lotka-Volterra}
  \end{subfigure}
  \hfill
  \begin{subfigure}[b]{0.49\textwidth}
    \includegraphics[width=\textwidth]{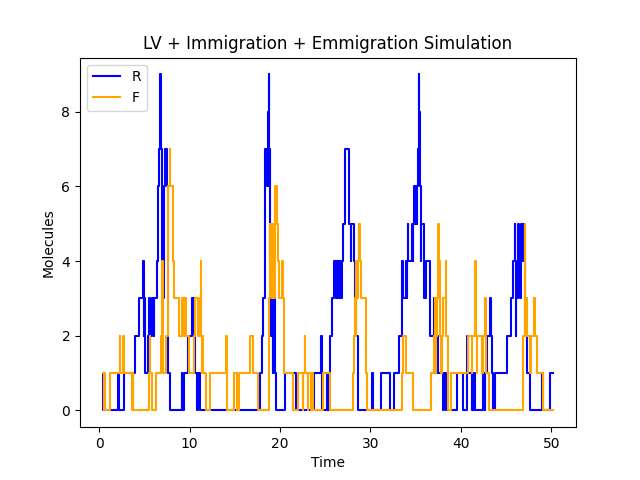}
    \caption{plus immigration and emigration rules.}\label{fig: Lotka-Volterra I&E}
  \end{subfigure}
  \caption{Stochastic CTMC simulations of Lotka-Volterra reaction system.}
\end{figure}

The Lotka-Volterra reaction system, $F+R\rightarrow 2.F,\ R\rightarrow 2.R,\ F\rightarrow\emptyset$,
is a classical model of interaction between two populations of preys (rabbits $R$) and predators (foxes $F$).
When interpreted using ODEs, the system exhibits oscillatory behaviour, with the populations of rabbits and foxes fluctuating over time (see Fig.~\ref{fig: LV_ODE}).
However, when interpreted by a CTMC, the system leads almost surely to the extinction of the predator or both populations (see Fig.~\ref{fig: Lotka-Volterra}).
This is due to the presence  in the stochastic semantics of an absorbing state 0
that is neither reached in the continuous semantics,
nor by Kurtz's limit mean CTMC trajectory obtained by making the volume tends to infinity~\cite{kurtz1970solutions}.
This is one reason why Kurtz's theorem should not be applied to this kind of counter-example.

On the other hand, one can remark that the addition of immigration and emigration reactions
prevents the system from reaching an absorbing state where one population is completely extinct.
By allowing for the possibility of individuals entering and leaving the system,
there is no absorbing state anymore and the populations exhibit sustained stochastic oscillations (see Fig.~\ref{fig: Lotka-Volterra I&E}).

\section{Preliminaries on Reaction Systems}

We recall here the standard notions of Chemical Reaction Networks (CRN) over molecular species~\cite{Feinberg77crt},
by using a more general terminology of Reaction Systems (RS) over mathematical variables,
in order to apply them to population dynamics or dynamical processes that have nothing to do with chemistry.
Apart from this change of terminology to reflect the generality of that formalism of reaction systems,
there is no difference between both formalisms of RSs and CRNs.

\begin{definition}
  A reaction $\rho$ over a set of variables $X$ is a triple $(R, P, \Lambda)$, also noted
\[\rho : \sum_i\l_i R_i \Rightarrow^{\Lambda} \sum_j r_j P_j\]
where $R, P$ are multisets of variables,
$R_i, P_j \in X$ are the elements of $R$, $P$ with non-zero {coefficient multiplicity} $\l_i,r_j\in \N^*$,
and $\Lambda:X\rightarrow\mathbb{R}^+$ is a rate function.

  A Reaction System (RS) over $X$ is a finite set of reactions over $X$.
\end{definition}

The $R_i$'s are called reactants and the $P_i$'s products.
A reaction with no reactant is called a \emph{synthesis} reaction.
A reaction with only one reactant and coefficient equal to 1 is called a \emph{unireactant} reaction.
Other reactions are \emph{polyreactant} reactions.
A polyreactant reaction with no product is called an \emph{annihilation} reaction.

\begin{notation}
  For a reaction $\rho$ and a variable $x\in X$, we write $\zeta_{\rho, x} = r_{\rho, x} - \l_{\rho, x}$
  the difference of coefficients between the production of product $x$ and the consumption of reactant $x$ in $\rho$.
\end{notation}

In this paper, we shall consider mass action law rate functions only.

\begin{definition}[Mass-Action Law]
  The mass action law rate of a reaction $(R,P,\Lambda)$ is of the form
  $\Lambda = \lambda\prod_i\prod_{k = 0}^{\l_i - 1} (R_i - k)$ for some rate constant $\lambda\in\R$.
\end{definition}

\begin{example}\label{first exemple}
  Mass action RS are written with just their rate constant, which can also be omitted if equal to 1, as in the reactions below:
\begin{center}
  {$\begin{aligned}
    A & \Rightarrow^2 A + B\\
    \_ & \Rightarrow B\\
    B + C & \Rightarrow^6 B + D\\
    D & \Rightarrow \_\\
\end{aligned}$}
\end{center}
\end{example}

\begin{definition}[Associated ODE]
The Ordinary Differential Equation (ODE) of an RS associates to each variable $x\in X$ the differential equation
\[\frac{\d x}{\d t} = \sum_\rho \Lambda_\rho(X)\zeta_{\rho, x}\]
{We denote by $\sigma$ the unique solution of the system starting at the initial state $\sigma(0)$, and {$\sigma_x$ its projection on the variable $x$}}
\end{definition}

\begin{example}
  The ODE associated with the RS of Ex.~\ref{first exemple} is
\begin{center}
  {$\begin{aligned}
    \frac{\d A}{\d t} & = 0 &
    \frac{\d B}{\d t} & = 2A + 1\\
    \frac{\d C}{\d t} & = -6BC &
    \frac{\d D}{\d t} & = 6BC - D\\
  \end{aligned}$}
\end{center}
\end{example}

\begin{definition}[Associated Continous Time Markov Chain]
  The Continous Time Markov Chain (CTMC) associated to an RS is defined by the following transition probabiility
  \[\P{X(t+\Delta t) = j | X(t) = i} \sim_{\Delta t\to 0} q_{ij}\Delta t\]
  where
  \[q_{X(t), X(t)+\zeta} = \sum_{\rho : r_\rho - l_\rho = \zeta} \Lambda_\rho(X(t))\]
\end{definition}

We shall make use of the following properties shown in~\cite{kurtz1970solutions} of the CTMC associated to an RS.

\begin{proposition}\label{eq: diff esp}
  The mean of $X(t)$ at each point of time follows the equation 
  \[\frac{\d \E{X(s)}}{\d t} = \sum_{\rho}\E{\Lambda_\rho(X(s))}\zeta_\rho\]
\end{proposition}

\begin{definition}[Generator]
  The generator of the CTMC is the function $A$ defined for each function $f$ by
  \[Af(x) = \sum_\rho{\Lambda_\rho}(x)(f(x+\zeta_\rho) - f(x))\]
\end{definition}

\begin{proposition}\label{eq: esp fct}
  $f(X(t)) - f(X(0)) - \int_{0}^{t}Af(X(s))\d s$
  is a martingale and so

  \noindent
  $\E{f(X(t))} = \E{f(X(0))} + \int_{0}^{t} \E{Af(X(s))}\d s.\hfill \displaystyle{\frac{\d\E{f(X(s))}}{\d t} = \E{Af(X(t))} }$
\end{proposition}

When $\zeta_{\rho, x} = 0$ the reaction $\rho$ has no direct influence on the variable $x$,
and one can focus on reactions where $\zeta_{\rho, x} \neq 0$ only.

\section{Stoichiometric Influence and Modification Graph}\label{sec:simg}

In this section, we introduce a special kind of influence graph associated to a reaction system,
which will be used to define our structural conditions ensuring the equality
between the mean stochastic and differential interpretations of a given RS.

\begin{definition}[SIMG]
The Stoichiometric Influence and Modification Graph (SIMG) of an RS is a bipartite multigraph
with variables and reactions as vertices,
with $\l_{\rho, x}$ arcs from variable $x$ to reaction $\rho$,
and one additional arc from $\rho$ to $x$ if $\zeta_{\rho, x} \neq 0$.
\end{definition}

\begin{example}\;
  THe SIMG of the RS of Ex.~\ref{first exemple} is the following:

    \centerline{\begin{tikzpicture}
	\begin{pgfonlayer}{nodelayer}
		\node [style=Species] (0) at (-1, 4) {$A$};
		\node [style=Species] (1) at (0, 0) {$B$};
		\node [style=Species] (2) at (-1.75, -2) {$C$};
		\node [style=Species] (3) at (2, -4) {$D$};
		\node [style=Reaction] (4) at (-1, 2) {};
		\node [style=Reaction] (5) at (2, 2) {};
		\node [style=Reaction] (6) at (2, -2) {};
		\node [style=Reaction] (7) at (-1.75, -4) {};
	\end{pgfonlayer}
	\begin{pgfonlayer}{edgelayer}
		\draw [style=Arrow, bend left=330, looseness=0.75] (3) to (7);
		\draw [style=Arrow, bend right, looseness=0.75] (7) to (3);
		\draw [style=Arrow] (0) to (4);
		\draw [style=Arrow] (4) to (1);
		\draw [style=Arrow] (5) to (1);
		\draw [style=Arrow] (1) to (6);
		\draw [style=Arrow, bend right, looseness=0.75] (2) to (6);
		\draw [style=Arrow, bend left=330, looseness=0.75] (6) to (2);
		\draw [style=Arrow] (6) to (3);
	\end{pgfonlayer}
\end{tikzpicture}}
  
\end{example}

Once this directed graph is constructed, the directed acyclic graph (DAG) of its strongly connected components (SCCs) can be created.
This DAG provides a hierarchy and depth for each node (reactions and variables) within the graph.
In this graph, each node's ancestors are all the SCCs (including itself) higher in the hierarchy and leading to it.

\begin{definition}[Ancestors]
  A variable $y\in X$ is called an ancestor of a variable $x\in X$
  if there exists a path in the SIMG from $y$ to $x$.
  Especially, $x$ is its own ancestor.
\end{definition}

\begin{definition}[Depth]
{The depth of a SCC is defined recursively as follows:
\begin{itemize}
    \item All SCCs with no arc leading to them are of depth $0$.
    \item Any other SCC is at depth $n+1$ where $n$ is the maximum depth of the SCCs leading to it.
\end{itemize}
By abuse of notation, we call the depth of a vertex $y$, the depth of the SCC of $y$.}
\end{definition}

\begin{example}\label{non trivial exemple}
  We show in the next section that the SIMG of the following RS satisfies graphical conditions
  that entail the equality between the deterministic ODE solution and the mean CTMC behaviour.
  Simulation traces are shown in Fig.~\ref{non trivial exemple plot}.
  The actual difference between the mean CTMC trace and the determinist ODE solution is shown in Fig.~\ref{non trivial exemple diff},
  using a scale that shows that the difference is indeed neglectable.

  \begin{center}
  \begin{tabular}{ccc}
    {$\begin{aligned}
      A + B & \Rightarrow A + B + C\\
      A + B & \Rightarrow A + B + D\\
      C & \Rightarrow D\\
    \end{aligned}$} &\ \ \ \ \ \ \ \ \ \ \ \ \ \ \ &
  \begin{tikzpicture}
	\begin{pgfonlayer}{nodelayer}
		\node [style=Species] (0) at (-2, 4) {$A$};
		\node [style=Species] (1) at (2, 4) {$B$};
		\node [style=Species] (2) at (-2, -1) {$C$};
		\node [style=Species] (3) at (2, -1) {$D$};
		\node [style=Reaction] (4) at (-2, 1.5) {};
		\node [style=Reaction] (5) at (2, 1.5) {};
		\node [style=Reaction] (6) at (2, 1.5) {};
		\node [style=Reaction] (7) at (0, -1) {};
	\end{pgfonlayer}
	\begin{pgfonlayer}{edgelayer}
		\draw [style=Arrow] (0) to (4);
		\draw [style=Arrow] (0) to (6);
		\draw [style=Arrow] (1) to (6);
		\draw [style=Arrow] (1) to (4);
		\draw [style=Arrow] (4) to (2);
		\draw [style=Arrow, bend left] (2) to (7);
		\draw [style=Arrow] (7) to (3);
		\draw [style=Arrow] (6) to (3);
		\draw [style=Arrow, bend left] (7) to (2);
	\end{pgfonlayer}
\end{tikzpicture} \\
  \end{tabular}
  \end{center}

  \begin{figure}[h!]
    \begin{subfigure}[b]{0.49\textwidth}
      \includegraphics*[width=\textwidth]{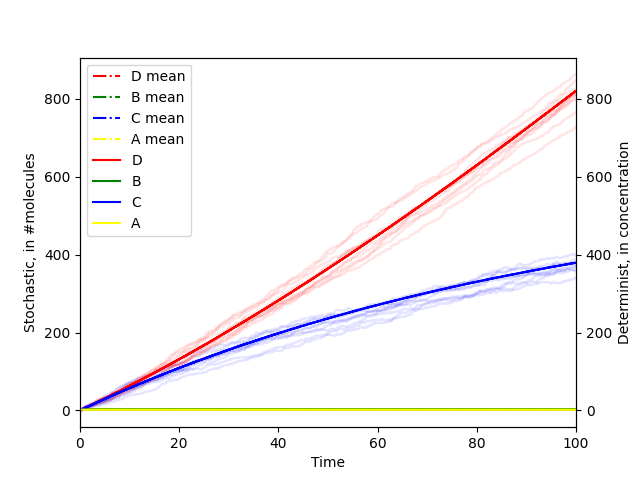}
      \caption{Stochastic mean trace and deterministic solution}\label{non trivial exemple plot}
    \end{subfigure}
    \hfill
    \begin{subfigure}[b]{0.49\textwidth}
      \includegraphics*[width=\textwidth]{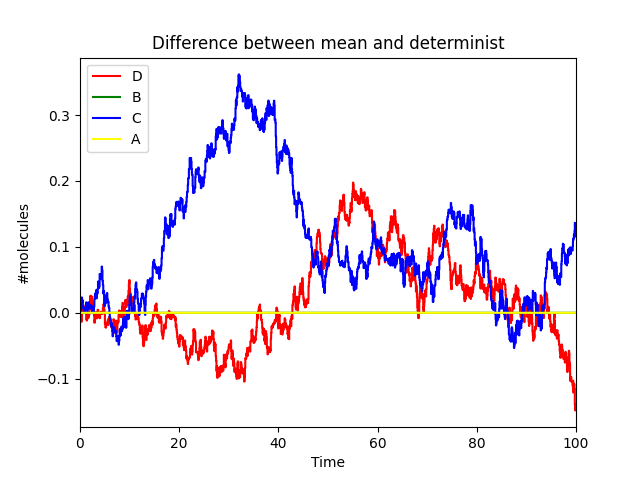}
      \caption{Difference between mean stochastic and determinist solutions}\label{non trivial exemple diff}
    \end{subfigure}
    \caption{Simulation traces for Ex.~\ref{non trivial exemple} based on 10000 stochastic traces.}\label{fig: non trivial exemple}
  \end{figure}

\end{example}

\section{Main Theorem}\label{sec:theorem}

\begin{theorem}\label{thm}
  For an RS with the mass-action law, if each reaction is
  \begin{enumerate}
    \item either at most unireactant,
    \item or polyreactant and has all reactants with disjoint SIMG ancestors,
  \end{enumerate}
  then the mean stochastic behaviour, starting from a fixed initial state, is equal to the differential solution with the same initial state.
\end{theorem}
\begin{proof}
  Let us examine the different types of reactions that can occur in the RS\@.
  As a reminder, they have the form:
  \[\rho : \sum_i\l_{\rho, R_i} R_i \Rightarrow^{\Lambda_\rho} \sum_j r_{\rho, P_j} P_j\]
  A first observation is that each variable is its own ancestor,
  which means it cannot have multiple arcs leading to the same reaction.
  Hence, for any reaction $\rho$ and variable $x$, we have $\l_{\rho, x} \in \{0, 1\}$.

  Therefore, for a specific variable $x$, reactions with $\zeta_{\rho, x} \neq 0$ can be categorized into the following five sets according to the reactants and their ancestors:
  \begin{enumerate}
    \item $\pi_0{(x)} = \{ \rho : \_ \Rightarrow^{\Lambda_\rho} r_{\rho, x} x + \sum_j r_{\rho, P_j} P_j\mid r_{\rho, x} > 0\}$
    \item $\pi_1{(x)} = \{ \rho : x \Rightarrow^{\Lambda_\rho} r_{\rho, x} x + \sum_j r_{\rho, P_j} P_j\mid r_{\rho, x} \neq 1\}$
    \item $\pi_1'{(x)} = \{ \rho : R \Rightarrow^{\Lambda_\rho} r_{\rho, x} x + \sum_j r_{\rho, P_j} P_j\mid r_{\rho, x} > 0, R \neq x\}$
    \item $\pi_2{(x)} = \{ \rho : x + \sum_i R_i \Rightarrow^{\Lambda_\rho} r_{\rho, x} x + \sum_j r_{\rho, P_j} P_j\mid r_{\rho, x} \neq 1, R_i \neq x\}$ and $x$ and all $R_i$ have disjoint ancestors.
    \item $\pi_2'{(x)} = \{ \rho : \sum_i R_i \Rightarrow^{\Lambda_\rho} r_{\rho, x} x + \sum_j r_{\rho, P_j} P_j\mid r_{\rho, x} > 0, R_i \neq x\}$ and $R_i$ have disjoint ancestors
  \end{enumerate}
  
  It is worth noting that if $\rho \in \pi_2{(x)}$, then reactant $R_0$ is an ancestor of both $x$ and $R_0$ itself.
  However, this violates the disjoint ancestors property condition.
  Therefore, we conclude that $\pi_2{(x)} = \emptyset$.

  Now, we can express the set of differential equations for the mean value of each variable:
  \begin{align*}
    \frac{\d\E{x}}{\d t} & = \sum_{\rho\in \pi_0{(x)}} \lambda_\rho\zeta_{\rho, x}\\
        & + \sum_{\rho\in \pi_1{(x)}} \lambda_\rho\E{x}\zeta_{\rho, x} + \sum_{\rho\in \pi_1'{(x)}} \lambda_\rho\E{R}\zeta_{\rho, x}\\
        & + \sum_{\rho\in \pi_2'{(x)}} \lambda_\rho\E{\prod_i R_i}\zeta_{\rho, x}\\
  \end{align*}
  Now, by applying Lemma~\ref{lemmaProd} proven below, the differential equation becomes:
  \begin{align*}
    \frac{\d\E{x}}{\d t} & = \sum_{\rho\in \pi_0{(x)}} \lambda_\rho\zeta_{\rho, x}\\
        & + \sum_{\rho\in \pi_1{(x)}} \lambda_\rho\E{x}\zeta_{\rho, x}  + \sum_{\rho\in \pi_1'{(x)}} \lambda_\rho\E{R}\zeta_{\rho, x}\\
        & + \sum_{\rho\in \pi_2'{(x)}} \lambda_\rho\prod_i \E{R_i}\zeta_{\rho, x}\\
  \end{align*}

  This is exactly the same set of differential equations satisfied by the differential solution $\sigma$.
  Furthermore, we have $\E{X}(0) = \sigma(0)$.
  Therefore, we conclude that $\E{X} = \sigma$.

\end{proof}
{
\begin{lemma}\label{lemmaSCC}
    Under the hypotheses of the theorem, a reaction is either alone in its SCC, or all the reactions of its SCC are unireatant reactions.
\end{lemma}
\begin{proof}
By simple inspection of the different cases.
\end{proof}}

\begin{lemma}\label{lemmaProd}
    Under the hypotheses of the theorem,
    for all variables $x_i$'s with no common ancestors, the following equality holds:
    \[\E{\prod_i x_i} = \prod_i\E{x_i}\]
  \end{lemma}
\begin{proof}
  Since the variables $x_i$'s have no common ancestors, {they are not modified by a same reaction}.
  Therefore, for any reaction $\rho$, we have $\zeta_{\rho, x_i}\zeta_{\rho, x_j} = 0$ for any $x_i$ and $x_j$ with no common ancestor {and no $x_i$ and $x_j$ are in the same SCC}.
  We prove the result by induction on the maximum depth of the variables occurring in the product.
  
  \textit{Base case:} At depth $0$, {for any $i$,} the sets $\pi_0{(x_i)}$ and $\pi_2{'(x_i)}$ are empty {otherwise the SIMG would have one of this patterns $\rho\rightarrow x_i$ or $R_1, R_2 \Rightarrow \rho\rightarrow x_i$ with $\rho$ in its own SCC, using Lemma \ref{lemmaSCC}, and so $x_i$ would have a depth $>0$}.
  
  {By using the equation from Proposition \ref{eq: esp fct},} we have:

  \begin{align*}
      \frac{\d\E{\prod_{i\in I} x_i}}{\d t} &= \sum_{j\in I}\sum_{\rho\in\pi_1{(x_j)}}\lambda_\rho\E{\prod_{i\in I}x_i}\zeta_{\rho, x_j} + \sum_{\rho\in\pi_1'{(x_j})}\lambda_\rho\E{R\prod_{i\neq j}x_i}\zeta_{\rho, x_j} \\
    \end{align*}

    For the product of means, we have:
    \begin{align*}
      \frac{\d\prod_{i\in I} \E{x_i}}{\d t} & = \sum_{j\in I}\sum_{\rho\in\pi_1{(x_j)}}\lambda_\rho\prod_{i\in I}\E{x_i}\zeta_{\rho, x_j} + \sum_{\rho\in\pi_1'{(x_j)}}\lambda_\rho \E{R}\prod_{i\neq j}\E{x_i}\zeta_{\rho, x_j} \\
    \end{align*}

    All {variables} appearing in those equations are of depth at most $0$.
    {We can rewrite those equation for any variable in each SCC of the $x_i$'s.
    This gives us a whole system of equations with one equation for each possible product of variables. Since only those products appear in the system, it has a unique solution at a fixed initial state.}
    The system of equations for the mean of the product and the product of means have the same structure.
    Since they have the same initial conditions, they are equivalent.
    Therefore, $\E{\prod_i x_i} = \prod_i \E{x_i}$ holds for any product of variables of depth $0$.

    \textit{Inductive step:}
    Let us assume {$\triangle$} that the result holds for any product of variables with a maximum depth of $n$,
    and prove it for products of variables with a maximum depth of $n+1$.
    We  use induction on the $n+2$-tuple $(d_{n+1}, \ldots, d_0)$ where $d_i$ represents the number of variables of the product at depth $i$, with lexicographic order.
    \begin{itemize}[label=$\star$]
      \item \textit{Base case:}
      If there are no variable at depth $n+1$,
      then the induction hypothesis $\triangle$ gives the result for products of variables with a maximum depth of $n$.
      Thus, the result holds.
      \item \textit{Induction step:}
      Assume that the result holds for any product with at most $k$ variables of depth exactly $n+1$ (and all others at most $n$).
      Considering a product with exactly $k+1$ variables of depth exactly $n+1$.
      We denote its $n+2$-tuple by $\tau = (k+1, d_n,\ldots, d_0)$.
      Consider the differential equation system for the product $\prod_{i\in I} x_i$:
      \begin{align*}
        \frac{\d\E{\prod_{i\in I} x_i}}{\d t} &= \sum_{j\in I}\sum_{\rho\in\pi_0{(x_j)}}\lambda_\rho\E{\prod_{i\neq j}x_i}\zeta_{\rho, x_j}\\
        &\quad + \sum_{\rho\in\pi_1{(x_j)}}\lambda_\rho\E{\prod_{i\in I}x_i}\zeta_{\rho, x_j} \\
        &\quad + \sum_{\rho\in\pi_1'{(x_j)}}\lambda_\rho\E{R\prod_{i\neq j}x_i}\zeta_{\rho, x_j}\\
        &\quad + \sum_{\rho\in\pi_2'{(x_j)}}\lambda_\rho\E{\prod_\l R_\l\prod_{i\neq j} x_i}\zeta_{\rho, x_j}
      \end{align*}

      For $\rho\in\pi_0{(x_j)}$, the tuple $\tau_{j, \rho} = \tau[d(x_j) \leftarrow \tau(d(x_j)) - 1]$ is such that $\tau_{j, \rho} < \tau$.
      Then, the induction hypothesis applies to the product $\E{\prod_{i\neq j} x_i} = \prod_{i\neq j}\E{x_i}$.

      Since for $\rho\in\pi_2'{(x_j)}$, using Lemma \ref{lemmaSCC}, the $R_\l$ are higher in the hierarchy than $x_j$, and they have no common ancestors with $x_i$ for $i\neq j$. So we have $\tau_{j, \rho} < \tau$ and thus $\E{\prod_\l R_\l\prod_{i\neq j} x_i} = \prod_\l\E{R_\l}\prod_{i\neq j}\E{x_i}$.
      The same resoning applies to $\rho\in\pi_1'{(x_j)}$ such that $R$ is not from the SCC of $x_j$. In this case, $R$ is higher and independant from the other variables, and $\E{R\prod_{i\neq j} x_i} = \E{R}\prod_{i\neq j}\E{x_i}$.

      It remains to consider $\rho\in\pi_1'{(x_j)}$ such that $R$ is part of the SCC of $x_j$, and also $\rho\in\pi_1{(x_j)}$ which is the same but with $R=x_j$.
      In this case, we can write the system of differential equations for $R\prod_{i\neq j} x_i$ for all $R$ of the SCC\@.
      Since we have the same number of equations as unknown functions and the initial values are fixed, we conclude that $\E{R\prod_{i\neq j}x_i}(0) = \left(\E{R}\prod_{i\neq j}\E{x_i}\right)(0)$,
      which implies $\E{R\prod_{i\neq j}x_i} = \E{R}\prod_{i\neq j}\E{x_i}$ for the considered product.
    \end{itemize}
\end{proof}

\section{SIMG Ancestor Condition Checking Algorithm}\label{sec:algo}

\begin{algorithm}[ht]
  \caption{Checking of the graphical condition}\label{algo}
  \SetKwData{Sccs}{sccs}\SetKwData{Scc}{scc}\SetKwData{Ancestors}{ancestors}\SetKwData{x}{x}\SetKwData{r}{$\rho$}
  \SetKwFunction{Depth}{depth}\SetKwFunction{Add}{add}
  \SetKwSwitch{Match}{Case}{Other}{match}{:}{$\bullet$}{$\_$}{}{}
  \SetKwInOut{Input}{input}\SetKwInOut{Output}{output}

  \Input{SIMG of an RS of size $n$ (number of variables and reactions)}
  \Output{Is the graphical condition of Thm.~\ref{thm} verified?}
  \BlankLine{}
  \Sccs{} calculus \tcp*[r]{in $O(n)$}
  Sort \Sccs{} with \Depth{\Scc1} $\leq$ \Depth{\Scc2} $\leq\ldots$\;
  \ForEach{\Scc$\in$ \Sccs}{
    Check \Scc is unireactant \tcp*[r]{in $O(|\Scc|)$}\nllabel{checking algo unireactant condition}
    \Ancestors$\leftarrow$ \Scc\;\nllabel{checking algo scc ancestors initialisation}
    \ForEach{variable \x $\in$ \Scc}{\label{cheking algo for variables in scc}
      \ForEach{reaction \r with \x as SIMG successor}{
        \Match{reactants of \r with}{
          \lCase{$\emptyset$}{do nothing}
          \lCase{$\{R\}$}{\label{checking algo add ancestors unireactant} \Ancestors.\Add{$R$.\Ancestors}}
          \Case{$\{R_1\ldots R_k\}$}{\label{checking algo polyreactant case}
            \eIf{$\forall i<j, R_i$.\Ancestors$\cap R_j$.\Ancestors$ = \emptyset$}{\label{checking algo disjoint ancestors}
              \Ancestors.\Add{$\cup_{i = 1} ^k R_i$.\Ancestors}\label{checking algo add ancestors poly}
            }{
              \Return{False}
            }
          }
        }
      }
    }
    \ForEach{variable \x $\in$ \Scc}{
      \x.\Ancestors$\leftarrow$ \Ancestors\;
    }
  }
  \Return{True}
  \BlankLine{}

\end{algorithm}

The construction of the SIMG of an RS can be easily computed from the reaction graph and is of linear size.
Now, Alg.~\ref{algo} provides a means to verify whether the graphical condition is fulfilled given the SIMG of an RS.
By simple inspection of the loops, one can show:
\begin{proposition}(termination and complexity)
Alg.~\ref{algo} terminates in $O(nd^2)$ steps where $n$ is the number of variables and reactions
and $d$ is the maximal number of reactants for a reaction.
\end{proposition}

\begin{proposition}(correctness)
Alg.~\ref{algo} checks the condition of Thm.~\ref{thm}.
\end{proposition}
\begin{proof}
  Let us prove its correctness by first showing that for each variable $x$,
  $x.ancestors$ are indeed the real ancestors of $x$.
  Then we'll see why this implies the correctness of the algorithm.

  First, it is clear that all members of a SCC have the same ancestors.
  
  Now, let us prove by induction on the depth of SCCs that $x.ancestors$ is actually the set of ancestors of $x$.
  
  \begin{itemize}
    \item At depth $0$:
      \begin{itemize}
        \item $\supseteq$ The only ancestors of any $x$ at depth $0$ are the variables from its SCC\@.
          They are added in Alg.~line~\ref{checking algo scc ancestors initialisation}
        \item $\subseteq$ The only reactions leading to any $x$ (in the SIMG) are the ones from its SCC\@.
          Then only Alg.~line~\ref{checking algo add ancestors unireactant} adds ancestors but it is nothing else than the SCC itself.
      \end{itemize}
      So ancestors at depth 0 are well computed.
    \item At depth $n$:
      \begin{itemize}
        \item $\supseteq$ If $A$ os an ancestor of $x$
          \begin{itemize}
            \item Either $A$ is in the same SCC than $x$ and is added in Alg.~line~\ref{checking algo scc ancestors initialisation}
            \item Either it is not and there exists a path in the SIMG such that $A \to \ldots\to R\to\square\to x'$, with some $R$ such that $depth(R)<n$ and $x'$ is in the same SCC than $x$.
              Then by induction hypothesis, $A.ancestors$ is well formed.
              So $R\in A.ancestors$ and $A.ancestors$ is adding to $x'.ancestors = x.ancestors$ when considering $x'$ in Alg.~line~\ref{cheking algo for variables in scc}.
          \end{itemize}
          Then all ancestors of $x$ are in $x.ancestors$.
        \item $\subseteq$ If $A$ is in $x.ancestors$.
          \begin{itemize}
            \item Either it is added with Alg.~line~\ref{checking algo add ancestors poly}
              and was in some $R_i.ancestors$ such that $R_i\to\square\to x'$ with $x'$ in the same SCC than $x$.
              So $depth(R_i) < n$ and then by induction hypothesis $A$ is an ancestor of $R_i$ and so of $x'$ and of $x$.
            \item Either it is added with Alg.~line~\ref{checking algo add ancestors unireactant} and so
              $R\to\square\to x'$ with $x'$ in the SCC of $x$ and $A$ in $R.ancestors$.
              The only way to be adde is then with $R.ancestors\neq\emptyset$ which can't happend with $R$ in the same SCC\@.
              So $depth(R)<n$ and the same reasoning as just before with induction hypothesis and transitivity of ancestors leads to $A$ as an ancestor of $x$.
          \end{itemize}
          So all variables in $x.ancestors$ are true ancestors of $x$.
      \end{itemize}
      Therefore $x.ancestors$ is well computed for each $x$ at depth $n$.
  \end{itemize}
  
  By the induction principle, it is the case for any depth.
  Finally, checking the condition of Thm.~\ref{thm} is done by Alg.~line~\ref{checking algo unireactant condition} and line~\ref{checking algo disjoint ancestors}.
\end{proof}

\section{Examples}

\begin{example}\label{non trivial exemple 2}
  One can check that the following RS satisfies the condition of Thm.~\ref{thm} on all variables.
  Fig.~\ref{fig: non trivial exemple 2} shows the equality between the ODE trace and the mean CTMC trace, together with their high variance.
  \begin{center}\footnotesize
  \begin{tabular}{ccc}
    {$\begin{aligned}
      A+B & \Rightarrow A + B + C \\
      C & \Rightarrow D\\
      D & \Rightarrow \_\\
      D & \Rightarrow D + E\\
      E & \Rightarrow E + F + G\\
      E & \Rightarrow \_\\
      F & \Rightarrow \_\\
      F & \Rightarrow F + C\\
      G & \Rightarrow H\\
      E + I & \Rightarrow E + I + H\\
      H & \Rightarrow \_\\
      D + J & \Rightarrow D + J + K\\
      K & \Rightarrow \_\\
      K + L & \Rightarrow K + L + M\\
      \_ & \Rightarrow M\\
      M & \Rightarrow \_\\
    \end{aligned}$}&\ \ \ \ \ \ \ \ \ \ \ \ \ \ \ &
    \begin{tikzpicture}
	\begin{pgfonlayer}{nodelayer}
		\node [style=Species] (0) at (-3.25, 2) {};
		\node [style=Reaction] (1) at (-1.25, 2.5) {};
		\node [style=Species] (2) at (0.75, 2) {};
		\node [style=Reaction] (3) at (1.25, 0) {};
		\node [style=Species] (4) at (0.75, -2) {};
		\node [style=Reaction] (5) at (-1.25, -2.75) {};
		\node [style=Species] (6) at (-3.25, -2) {};
		\node [style=Reaction] (7) at (-3.75, 0) {};
		\node [style=Reaction] (8) at (-3.25, 3.75) {};
		\node [style=Reaction] (9) at (1.75, -3.5) {};
		\node [style=Species] (10) at (-4.75, 4.75) {};
		\node [style=Species] (11) at (-1.75, 4.75) {};
		\node [style=Species] (12) at (1.75, -5) {};
		\node [style=Species] (13) at (-1.25, -4.5) {};
		\node [style=Reaction] (14) at (0.25, -5) {};
		\node [style=Reaction] (15) at (2.75, 1) {};
		\node [style=Species] (16) at (2.75, -0.75) {};
		\node [style=Species] (17) at (5.75, -0.5) {};
		\node [style=Reaction] (18) at (4.25, -1.75) {};
		\node [style=Species] (19) at (5.75, -3) {};
		\node [style=Reaction] (20) at (7, -1.75) {};
		\node [style=Reaction] (21) at (-5, -2.25) {};
		\node [style=Reaction] (22) at (0.75, 3.5) {};
		\node [style=Species] (23) at (4.75, 2) {};
		\node [style=Species] (24) at (2.75, -2) {};
		\node [style=Reaction] (25) at (-0.5, -1) {};
		\node [style=Reaction] (27) at (4.25, 0) {};
		\node [style=Reaction] (28) at (7, -4.25) {};
		\node [style=Reaction] (30) at (3.5, -5) {};
	\end{pgfonlayer}
	\begin{pgfonlayer}{edgelayer}
		\draw [style=Arrow, bend left] (0) to (1);
		\draw [style=Arrow] (1) to (2);
		\draw [style=Arrow] (2) to (3);
		\draw [style=Arrow] (3) to (4);
		\draw [style=Arrow] (4) to (5);
		\draw [style=Arrow] (5) to (6);
		\draw [style=Arrow] (6) to (7);
		\draw [style=Arrow] (7) to (0);
		\draw [style=Arrow] (5) to (13);
		\draw [style=Arrow, bend left] (13) to (14);
		\draw [style=Arrow] (14) to (12);
		\draw [style=Arrow] (4) to (9);
		\draw [style=Arrow] (9) to (12);
		\draw [style=Arrow] (8) to (0);
		\draw [style=Arrow] (10) to (8);
		\draw [style=Arrow] (11) to (8);
		\draw [style=Arrow] (2) to (15);
		\draw [style=Arrow] (17) to (18);
		\draw [style=Arrow] (16) to (18);
		\draw [style=Arrow] (20) to (19);
		\draw [style=Arrow] (18) to (19);
		\draw [style=Arrow] (15) to (16);
		\draw [style=Arrow, bend left] (6) to (21);
		\draw [style=Arrow, bend left] (21) to (6);
		\draw [style=Arrow, bend right, looseness=1.25] (22) to (2);
		\draw [style=Arrow, bend right, looseness=1.25] (2) to (22);
		\draw [style=Arrow] (23) to (15);
		\draw [style=Arrow] (24) to (9);
		\draw [style=Arrow, bend left] (1) to (0);
		\draw [style=Arrow, bend right=330] (14) to (13);
		\draw [style=Arrow, bend right, looseness=0.75] (28) to (19);
		\draw [style=Arrow, bend right, looseness=0.75] (19) to (28);
		\draw [style=Arrow, bend right] (12) to (30);
		\draw [style=Arrow, bend right, looseness=0.75] (25) to (4);
		\draw [style=Arrow, bend right=15] (4) to (25);
		\draw [style=Arrow, bend left] (27) to (16);
		\draw [style=Arrow, bend left] (16) to (27);
		\draw [style=Arrow, bend left=330] (30) to (12);
	\end{pgfonlayer}
\end{tikzpicture}
  \end{tabular}
  \end{center}

    \begin{figure}[h]
      \centering
      \begin{subfigure}[b]{0.49\textwidth}
      \includegraphics*[width=\textwidth]{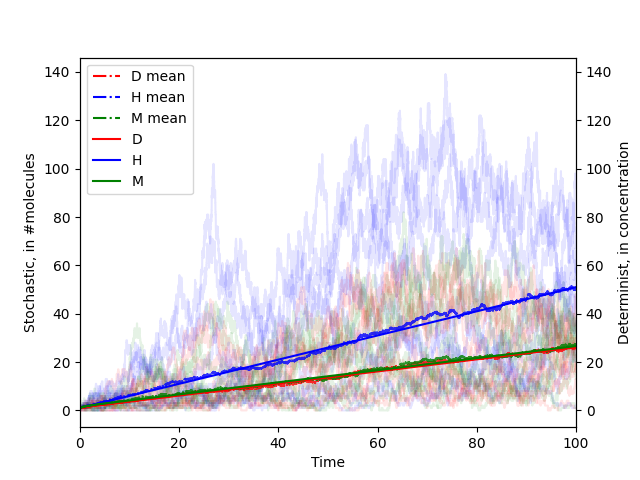}
      \caption{Stochastic mean and determinist solutions.}\label{non trivial exemple 2 plot}
      \end{subfigure}
      \hfill
      \begin{subfigure}[b]{0.49\textwidth}
        \centering
        \includegraphics*[width=\textwidth]{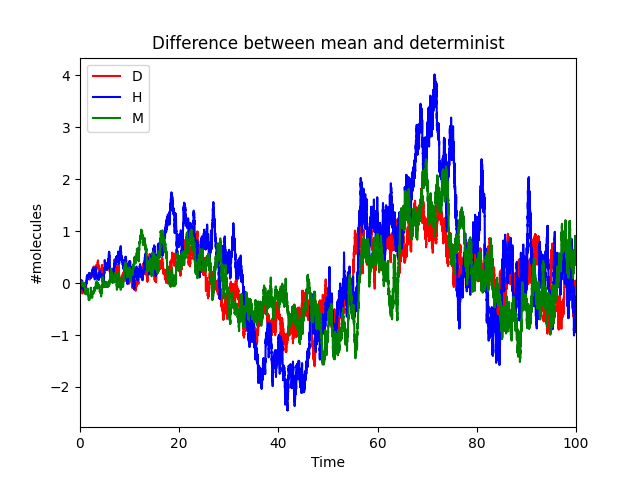}
        \caption{Difference between mean stochastic and determinist solutions.}\label{non trivial exemple 2 diff}
      \end{subfigure}
      \caption{ODE trace and 100 CTMC traces with mean in Ex.~\ref{non trivial exemple 2}.}\label{fig: non trivial exemple 2}
    \end{figure}

\end{example}

\begin{example}
  The following RS implements arithmetic multiplication~\cite{FLBP17cmsb}.
  Indeed, with $A$ and $B$ constant,
in the ODE semantics we get at stable state, $\E{C} = AB$.
Thm.~\ref{thm} applies on all variables and thus entails the correctness of the result in the mean stochastic interpretation as well:
  \begin{center}\begin{tabular}{ccc}
  {$\begin{aligned}
    A + B & \Rightarrow A + B + C\\
    C & \Rightarrow \_\\
  \end{aligned}$} &\ \ \ \ \ \ \ \ \ \ \ \ \ \ \ &
\begin{tikzpicture}
	\begin{pgfonlayer}{nodelayer}
		\node [style=Species] (0) at (-2, 3) {$A$};
		\node [style=Species] (1) at (2, 3) {$B$};
		\node [style=Species] (2) at (0, -1) {$C$};
		\node [style=Reaction] (3) at (0, 1.25) {};
		\node [style=Reaction] (4) at (2.5, -1) {};
	\end{pgfonlayer}
	\begin{pgfonlayer}{edgelayer}
		\draw [style=Arrow] (0) to (3);
		\draw [style=Arrow] (1) to (3);
		\draw [style=Arrow] (3) to (2);
		\draw [style=Arrow, bend right] (2) to (4);
		\draw [style=Arrow, bend left=330] (4) to (2);
	\end{pgfonlayer}
\end{tikzpicture} \\    
    \end{tabular}
    \end{center}

\end{example}

\begin{example}\label{counter_example}
  Let us consider the following RS with SIMG violating the conditions:

\begin{center}\begin{tabular}{ccc}
{$\begin{aligned}
  A + B & \Rightarrow A + B + C\\
  C & \Rightarrow \_ \\
  C & \Rightarrow C + D\\
  D & \Rightarrow \_ \\
  C + D & \Rightarrow C + D + E\\
  E & \Rightarrow \_ \\ 
  \end{aligned}$} &\ \ \ \ \ \ \ \ \ \ \ \ \ \ \ &
\begin{tikzpicture}
	\begin{pgfonlayer}{nodelayer}
		\node [style=Species] (0) at (-3.25, 3.75) {$A$};
		\node [style=Species] (1) at (0.75, 3.75) {$B$};
		\node [style=Species] (2) at (2.75, -0.25) {$D$};
		\node [style=Species] (3) at (-1.25, -0.25) {$C$};
		\node [style=Reaction] (4) at (0.75, -1.75) {};
		\node [style=Reaction] (6) at (-1.25, 1.75) {};
		\node [style=Reaction] (7) at (0.75, -0.25) {};
		\node [style=Species] (8) at (0.75, -3.75) {$E$};
		\node [style=Reaction] (9) at (5, -0.25) {};
		\node [style=Reaction] (10) at (2.5, -3.75) {};
		\node [style=Reaction] (11) at (-3.5, -0.25) {};
	\end{pgfonlayer}
	\begin{pgfonlayer}{edgelayer}
		\draw [style=Arrow] (0) to (6);
		\draw [style=Arrow] (1) to (6);
		\draw [style=Arrow] (6) to (3);
		\draw [style=Arrow] (3) to (4);
		\draw [style=Arrow] (3) to (7);
		\draw [style=Arrow] (7) to (2);
		\draw [style=Arrow] (2) to (4);
		\draw [style=Arrow] (4) to (8);
		\draw [style=Arrow, bend left=330] (9) to (2);
		\draw [style=Arrow, bend right] (2) to (9);
		\draw [style=Arrow, bend left] (8) to (10);
		\draw [style=Arrow, bend left] (10) to (8);
		\draw [style=Arrow, bend left] (11) to (3);
		\draw [style=Arrow, bend left] (3) to (11);
	\end{pgfonlayer}
\end{tikzpicture}
    \end{tabular}
\end{center}
  
  \begin{figure}[h!]
    \centering
    \includegraphics*[scale=0.45]{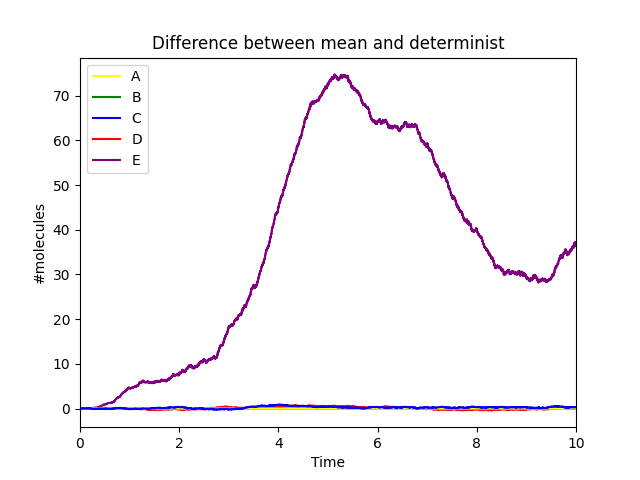}
    \caption{Difference between the mean over 100 iterations and the deterministic solution for Ex.~\ref{counter_example} with $A(0)=5$ and $B(0)=11$}\label{counter_example_figure}
  \end{figure}

  The difference of behaviours between the stochastic and deterministic interpretations is plotted in Fig.~\ref{counter_example_figure} for the different variables.
  Interestingly, one can observe no difference on all variables except E.
  This can be proven on the SIMG as follows.
Because the sub-graph until $D$ verifies the graphical condition,
we have $\E{A} = \sigma_A = a$, $\E{B} = \sigma_B = b$ and $\E{C}=\sigma_C$ and $\E{D} = \sigma_D$.

\noindent
Then $\displaystyle\frac{\d\sigma_C}{\d t} = ab + \sigma_C$ so, with $\sigma_C(0) = 0$ we have $\sigma_C = ab(1-\e^{-t})$.
Thus
\[\begin{aligned}
  \frac{\d\xi}{\d t}(t) = \frac{\d (\E{C^2} - \sigma_C^2)}{\d t}(t) & = \E{ab ({(C+1)}^2 - C^2) + C({(c-1)}^2 - C^2)}(t) \\ &\quad- (2(ab - \sigma_C)\sigma_C)(t)\\
                        & = (ab + \sigma_C -2\xi)(t)
\end{aligned}\]
and $\xi(t) = ab(1-\e^{-t})$.

The same reasoning leads for $\psi = \E{CD} - \sigma_C\sigma_D$ to
$\psi' + 2\psi = \xi$ and so $\psi(t) = \frac{ab}{2}(1 - 2\e ^{-t} + \e^{-2t})$.

Then $\phi = \E{E} - \sigma_E$ verifies $\phi' + \phi = \psi$
and so $\phi = \frac{ab}{2}(1 - 2t\e^{-t} - \e^{-2t})\neq 0$.

\end{example}

\section{Equality Property for Cosine Oscillatory RS}\label{sec:cos}

  The sine and cosine functions of time constitute a fundamental example of oscillator.
  A CRN implementing $sin(t)$ and $cos(t)$ with 4 molecular species was given in~\cite{FLBP17cmsb}
  for illustrating the Turing-completeness of continuous CRNs,
  and the compiler of computable mathematical functions in CRNs implemented in BIOCHAM\footnote{\url{http://lifeware.inria.fr/biocham}}~\cite{CFS06bi}.
  That automatically generated RS, called cosRS, uses the dual-rail encoding of negative values as the difference between the concentrations of two reactants
  with a fast annihilation reaction between both with rate $f$:
\begin{center}\small$\begin{aligned}
  A & \Rightarrow A + B \\
  B & \Rightarrow B + C \\
  C & \Rightarrow C + D \\
  D & \Rightarrow D + A \\
  A + C & \Rightarrow^f \_ \\
  B + D & \Rightarrow^f \_ \\
  \end{aligned}$
\end{center}

With initial conditions $A(0) = 1$ and $B(0) = C(0) = D(0) = 0$,
we have ${A-C} = \cos(t)$ and ${D-B} = \sin(t)$ in the ODE semantics of cosRS.
We show that the same holds for $\E{A-C}$ and $\E{D-B}$ in the stochastic CTMC semantics,
and does generalize to any initial conditions.

\begin{theorem}
{For any initial conditions, the mean stochastic traces of $A - C$ and $D-B$ are equal to their ODE traces in the ODE semantics of cosRS.}
\end{theorem}
\begin{proof}
  We just have to show that 
    the mean stochastic CTMC traces of $D - B$ and $A - C$ obey the differential equations of the ODE semantics of cosRS.
  
  Variables $A$ and $C$ represent the positive and negative parts of the same function, $A-C$,
while $B$ and $D$ represent similarly another function, $D-B$.

\noindent
By linearity, we have $\displaystyle\frac{\d\E{A-C}}{\d t} = \E{D-B}$ and $\displaystyle\frac{\d\E{D-B}}{\d t} = \E{C-A}$.

\noindent
Therefore, we get $\displaystyle\frac{\d^2\E{A-C}}{\d t^2} = -\E{A-C}$ and $\displaystyle\frac{\d^2\E{D-B}}{\d t^2} = -\E{D-B}$.

\noindent
\end{proof}

\section{Evaluation in \texttt{Biomodels} and Partial Approximations}

Alg.~\ref{algo} has been applied to the first 2000 models of the database BioModels~\cite{BioModels2020}.
The only models satisfying the condition did not contain polyomolecular reactions.
As mentioned in the introduction, the property for at most unimolecular reaction models is well-known
and our graphical condition for all variables does not apply to more models in Biomodels repository.
Nevertheless, even when the conditions of Thm.~\ref{thm} are not satisfied on the whole SIMG of the model,
they can be used to identify a subset of reactant variables for which the ODE solution is indeed equal to the mean CTMC behaviour.

The key observation is that the condition is additive with respect to the variables ordering based on their depth in the SIMG.
If variables do not violate the condition, then it is true for all previous variables considered in the algorithm.
Therefore, the algorithm only needs to focus on the partial portion of the graph where all ancestors satisfy the condition.
The result then holds for those variables and means nothing for the other variables.

A partial application of Thm.~\ref{thm} can thus be used to identify RS input variables
which may be non constant, but with provably equal differential and stochastic mean dynamics.
Independently of the more complex behaviours of the other variables,
the theorem guaranties that the ODE semantics provides the exact mean concentration of these input variables over time.

This is illustrated here in Ex.~\ref{counter_example} where the partial condition applies to variables A, B, C and D.
This explains why in Fig.~\ref{counter_example_figure}, variables A, B, C and D follow the determinist ODE solution in mean, unlike variable E.
{In BioModels, the statistics are however not very significant since many reactions have no mass-action law kinetics,
  or have non well-formed kinetics \cite{FGS15tcs}.}

\section{Conclusion and Perspectives}

Kurtz's theorem is the state-of-the-art result to relate the mean stochastic CTMC and deterministic ODE interpretations of a reaction system.
As a limit theorem however, it does not apply to bounded systems
and leads to wrong results if applied for instance to population dynamics models in finite space.

We have shown that the strong equality holds between the mean CTMC and deterministic ODE interpretations
independently of the kinetics and without any limit approximation,
provided some graphical conditions are satisfied on the structure the RS.
The interest of our result is that it is not a limit theorem and can thus be applied to bounded dynamical systems with low values for the variables,
such as models of gene regulatory networks if we think that one gene is in single copy in a cell.

The evaluation of our graphical conditions in the reaction models in Biomodels shows, on the one hand, that they are satisfied for all species,
only in models containing no polymolecular reactions,
but, on the other hand, that they can be used to identify in each model a subset of molecular species for which the property holds.
Interestingly, we have also shown that the basic oscillatory RS implementing the sine and cosine functions
enjoys the strong equality property for all species.

These results also leave some new questions open.
First, back to our first example,
can we determine the mean stochastic behaviour of the Lotka-Volterra system extended with immigration and emigration reactions
for which we have shown the existence of sustained stochastic oscillations ?
Second, can we generalize the reasoning done for the sine-cosine RS,
to prove the property for general RSs similarly involving annihilation reactions as sole degradation reactions ?
Third, our assumptions require no correlations between reactants,
while some variables may be balanced and yield the same result, they do not meet our condition.
These questions should deserve further investigations.

\subsection*{Acknowledgements} We are grateful to the reviewers of CMSB for their comments,
which where useful to improve the presentation of our results,
and to Mathieu Hemery for discussions, and for improving the parsing of SBML models in BIOCHAM.
This research was partially funded by ANR-20-CE48-0002 project Difference.
  
  %
%
%
\bibliographystyle{splncs04}
\bibliography{bilbio}

\end{document}